\documentclass[letterpaper,11pt]{article}
\usepackage{url}
\usepackage{amssymb,amsmath,amsfonts,amsthm}
\usepackage{graphicx}
\usepackage{fullpage}
\usepackage[usenames]{color}

\DeclareMathOperator*{\E}{\mathbb{E}}
\let\Pr\relax
\DeclareMathOperator*{\Pr}{\mathbb{P}}

\DeclareMathSymbol{\qedsymb} {\mathord}{AMSa}{"04}

\newcommand{\eps}{\varepsilon}
\renewcommand{\epsilon}{\varepsilon}

\newcommand{\ceil}[1]{\left\lceil #1 \right\rceil}
\newcommand{\floor}[1]{\left\lfloor #1 \right\rfloor}

\newcommand{\oct}{\quad\quad}                                   

\newcommand{\R}{\mathbb{R}}

\newcommand{\inprod}[1]{\langle #1 \rangle}

\newcommand{\CorollaryName}[1]{\label{cor:#1}}

\newcommand{\EquationName}[1]{\label{eq:#1}}
\newcommand{\FactName}[1]{\label{fact:#1}}
\newcommand{\LemmaName}[1]{\label{lem:#1}}

\newcommand{\RemarkName}[1]{\label{rem:#1}}
\newcommand{\SectionName}[1]{\label{sec:#1}}
\newcommand{\TheoremName}[1]{\label{thm:#1}}
\newcommand{\FigureName}[1]{\label{fig:#1}}

\newcommand{\Corollary}[1]{Corollary~\ref{cor:#1}}

\newcommand{\Equation}[1]{Eq.\:\eqref{eq:#1}}
\newcommand{\Fact}[1]{Fact~\ref{fact:#1}}
\newcommand{\Lemma}[1]{Lemma~\ref{lem:#1}}

\newcommand{\Remark}[1]{Remark~\ref{rem:#1}}
\newcommand{\Section}[1]{Section~\ref{sec:#1}}
\newcommand{\Theorem}[1]{Theorem~\ref{thm:#1}}
\newcommand{\Figure}[1]{Figure~\ref{fig:#1}}

\newtheorem{theorem}{Theorem}
\newtheorem{corollary}[theorem]{Corollary}
\newtheorem{conjecture}[theorem]{Conjecture}
\newtheorem{definition}[theorem]{Definition}

\newtheorem{fact}[theorem]{Fact}

\newtheorem{lemma}[theorem]{Lemma}
\newtheorem{remark}[theorem]{Remark}

\newcommand{\proofbelow}{3pt}
\newcommand{\afterproof}{\hfill $\blacksquare$ \par \vspace{\proofbelow}}
\newcommand{\aftersubproof}{\hfill $\Box$ \par \vspace{\proofbelow}}
\renewenvironment{proof}{\noindent\textbf{Proof.}\,}{\afterproof}

\newenvironment{proofof}[1]{\noindent\textbf{Proof} \,(of #1).\,}{\afterproof}

\renewcommand{\th}{\ifmmode{^{\textrm{th}}}\else{\textsuperscript{th}\ }\fi}

\newcommand{\eqdef}{\mathbin{\stackrel{\rm def}{=}}}
\newcommand{\comment}[1]{}

\begin{document}

\author{Daniel M. Kane\footnote{Stanford University, Department of
  Mathematics. \texttt{dankane@math.stanford.edu}. This work was done while the author was supported by an NSF Graduate Research Fellowship.}\oct
Jelani Nelson\footnote{Harvard
  University, School of Engineering and Applied Sciences. \texttt{minilek@seas.harvard.edu}. This work was done while the author was supported by a Xerox-MIT Fellowship, and in part by the
Center for Massive Data Algorithmics (MADALGO) - a center of the Danish National Research Foundation.
}}

\date{}

\title{Sparser Johnson-Lindenstrauss Transforms}

\maketitle

\begin{abstract}
We give two different and simple constructions for dimensionality reduction in $\ell_2$ via linear mappings that are sparse: only an $O(\eps)$-fraction of entries in each column of our embedding matrices are non-zero to achieve distortion $1+\eps$ with high probability, while still achieving the asymptotically optimal number of rows. These are the first constructions to provide subconstant sparsity for all values of parameters, improving upon previous works of Achlioptas (JCSS 2003) and Dasgupta, Kumar, and Sarl\'{o}s (STOC 2010).  Such distributions can be used to speed up applications where $\ell_2$ dimensionality reduction is used.
\end{abstract}

\section{Introduction}\SectionName{intro}
The Johnson-Lindenstrauss lemma states:

\begin{lemma}[JL Lemma {\cite{JL84}}]\LemmaName{jl-lemma}
For any integer $d>0$, and any $0<\eps,\delta<1/2$,
there exists a probability distribution on $k\times d$ 
real matrices for $k = \Theta(\eps^{-2}\log(1/\delta))$ 
such that for any $x\in\R^d$,

$$ \Pr_S((1-\eps)\|x\|_2\le \|Sx\|_2 \le (1+ \eps)\|x\|_2) >
1 - \delta .$$
\end{lemma}

Proofs of the JL lemma can be found in
\cite{Achlioptas03,AV06,BOR10,DG03,FM88,HIM12,JL84,KN10,Matousek08}.
The value of $k$ in the JL lemma is optimal \cite{JW13}
(also see a later proof in \cite{KMN11}).

The JL lemma is a key ingredient in the JL flattening theorem, which
states that any $n$ points in Euclidean space can be embedded into
$O(\eps^{-2}\log n)$ dimensions so that all pairwise Euclidean
distances are preserved up to $1\pm\eps$.
The JL lemma is a useful tool for speeding up solutions to several
high-dimensional problems: closest pair, nearest neighbor, diameter,
minimum spanning tree, etc.  It also speeds up some clustering and
string processing algorithms, and can further be used to reduce the
amount of storage required to store a dataset, e.g.\ in streaming
algorithms. Recently it has also found applications in approximate
numerical algebra problems such as linear regression and low-rank
approximation \cite{CW09, Sarlos06}. See
\cite{Indyk01,Vempala04} for further discussions on
applications.

Standard proofs of the JL lemma take a distribution over dense matrices
(e.g. i.i.d. Gaussian or Bernoulli entries), and thus performing the
embedding na\"{i}vely takes $O(k\cdot \|x\|_0)$ time where $x$ has
$\|x\|_0$ non-zero entries. Several works have devised other
distributions which give faster embedding times
\cite{AC09,AL09,AL13,HV11,KW11,Vybiral11}, but all these methods
require $\Omega(d\log d)$ embedding time even for sparse vectors (even when $\|x\|_0 = 1$). This
feature is particularly unfortunate in streaming applications, where a
vector $x$ receives coordinate-wise updates of the form $x\leftarrow
x + v\cdot e_i$, so that to maintain some linear
embedding $Sx$ of $x$ we should repeatedly calculate $Se_i$ during
updates.  Since $\|e_i\|_0 = 1$, even the na\"{i}ve $O(k\cdot
\|e_i\|_0)$ embedding time method is faster than these approaches.  

Even aside from
streaming applications, several practical situations give rise to
vectors with $\|x\|_0 \ll d$.  
For example, 
a common similarity measure for comparing text documents in data mining and
information retrieval is cosine
similarity \cite{TSK05}, which is approximately preserved under any JL
embedding. Here, a document is
represented as a bag of
words with the dimensionality  $d$ being the
size of the lexicon, and we usually would not expect any
single document to contain anywhere near $d$ distinct words (i.e., we
expect sparse vectors).  
In networking
applications, if $x_{i,j}$ counts bytes sent from source $i$ to
destination $j$ in some time interval, then $d$ is the total number of IP
pairs, whereas we would not expect most pairs of IPs to
communicate with
each other. In linear algebra applications, a rating matrix $A$ may
for example have $A_{i,j}$ as user $i$'s score for item $j$ (e.g.\
the Netflix matrix where columns correspond to movies), and we would
expect that most users rate only
small fraction of all available items.

One way to speed up embedding time in the JL lemma for sparse vectors
is to devise a
distribution over sparse embedding matrices. This was first
investigated in \cite{Achlioptas03}, which gave a JL distribution where
only one third of the entries of each matrix in its support was
non-zero, without increasing the number of rows $k$ from 
dense constructions.
Later, the works \cite{CCF04,ThorupZhang12} gave a distribution over matrices
with only $O(\log(1/\delta))$ non-zero entries per column, but the
algorithm for estimating $\|x\|_2$
given the linear sketch then relied on a
median calculation, and thus these schemes
did not provide an embedding into $\ell_2$. In several applications,
such as nearest-neighbor search \cite{HIM12} and approximate numerical
linear algebra \cite{CW09, Sarlos06}, an embedding into a normed space
or even $\ell_2$ itself
is required, and thus median estimators cannot be used. Median-based
estimators also pose a problem when one wants to learn classifiers in
the dimension-reduced space via stochastic gradient descent, since in
this case the estimator needs certain differentiability properties
\cite{WDLSA09}. In fact, the work of \cite{WDLSA09} 
investigated JL distributions over sparse matrices for this reason, in
the context of collaborative spam filtering.
The work \cite{DKS10} later analyzed the JL distribution in
\cite{WDLSA09} and showed
that it can be realized
where for each matrix in the support of the distribution, each column
has at most $s =
\tilde{O}(\eps^{-1}\log^3(1/\delta))$\footnote{We say $g =
   \tilde{\Omega}(f)$ when $g =
   \Omega(f/ \mathrm{polylog}(f))$, $g = \tilde{O}(f)$ when $g =
   O(f\cdot \mathrm{polylog}(f))$, and $g = \tilde{\Theta}(f)$ when $g
   = \tilde{\Omega}(f)$ and $g = \tilde{O}(f)$ simultaneously.}
   non-zero entries, thus speeding
up the embedding time to $O(s\cdot \|x\|_0)$.
This ``DKS construction''
requires $O(ds \log k)$ bits of random seed to sample a 
matrix from their distribution.  The work of \cite{DKS10} left open
two main directions: (1) understand the
sparsity parameter $s$ that can be achieved in a JL distribution,
and (2) devise a sparse JL transform distribution which requires few
random bits to
sample from, for streaming applications where storing a long random
seed requires prohibitively large memory.

The previous work \cite{KN10} of the current authors made progress on
both these questions by showing $\tilde{O}(\eps^{-1}\log^2(1/\delta))$
sparsity was achievable by giving an alternative analysis of the
scheme of \cite{DKS10} which also only required
$O(\log(1/(\eps\delta))\log d)$ seed length. The work of \cite{BOR10}
later gave a tighter analysis
under the assumption $\eps < 1/\log^2(1/\delta)$,
improving the sparsity and seed length further by $\log(1/\eps)$ and
$\log\log(1/\delta)$ factors in this case. In \Section{tight} we show
that the DKS scheme {\em requires} $s =
\tilde{\Omega}(\eps^{-1}\log^2(1/\delta))$, and thus a departure from
their construction is required to obtain better sparsity. For a
discussion of other previous work concerning the JL lemma see
\cite{KN10}.

\begin{figure}
\begin{center}
\begin{tabular}{cc}
\scalebox{0.4}{\includegraphics{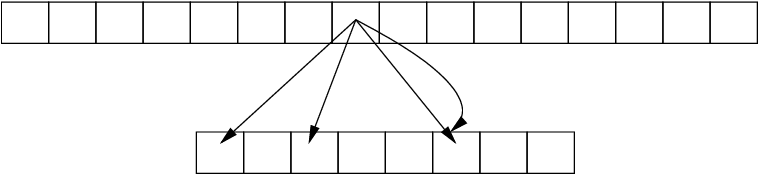}} &
\scalebox{0.4}{\includegraphics{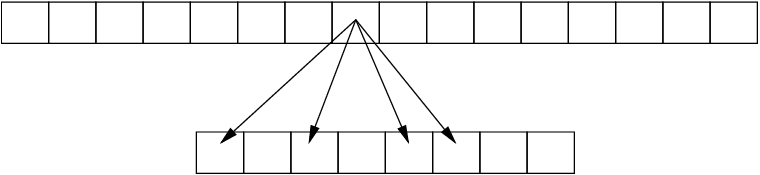}}\\
(a) & (b)
\end{tabular}
\begin{center}
\scalebox{0.4}{\includegraphics{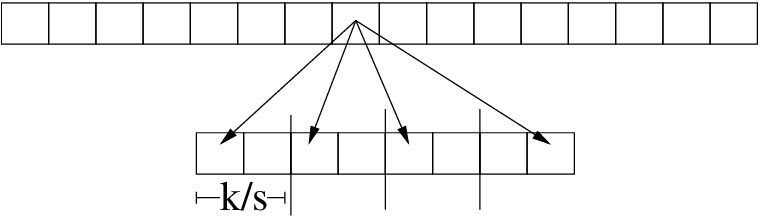}}\\
(c)
\end{center}
\caption{In all three constructions above, a vector in $\R^d$ is 
  projected down to $\R^k$. Figure (a) is the DKS construction in
  \cite{DKS10}, and the two constructions we give in this work are
  represented in (b) and (c).  The out-degree in each case is $s$, the
sparsity.}\FigureName{pictures}
\end{center}
\end{figure}

\bigskip

\paragraph{Main Contribution:} In this work, we give two new
constructions which achieve sparsity $s =
\Theta(\eps^{-1}\log(1/\delta))$ for $\ell_2$ embedding into optimal
dimension $k = \Theta(\eps^{-2}\log(1/\delta))$. This is the first sparsity
bound which is always $o(k)$ for the asymptotically optimal value of $k$ for all ranges of $\eps,\delta$. One of our
distributions can be sampled from using
$O(\log(1/\delta)\log d)$ uniform random bits.

\bigskip

It is also worth nothing that after the preliminary version of this work was published in \cite{KN12}, it was shown in \cite{NN13} that our bound is optimal up to an $O(\log(1/\eps))$ factor. That is, for any fixed constant $c>0$, {\em any} distribution satisfying \Lemma{jl-lemma} that is supported on matrices with $k = O(\eps^{-c}\log (1/\delta))$ and at most $s$ non-zero entries per column must have $s = \Omega(\eps^{-1}\log(1/\delta)/\log(1/\eps))$ as long as $k = O(d/\log(1/\eps))$. Note that once $k\ge d$ one can always take the distribution supported solely on the $d\times d$ identity matrix, giving $s=1$ and satisfying \Lemma{jl-lemma} with $\eps=0$.

We also describe variations on our constructions which achieve
sparsity $\tilde{O}(\eps^{-1}\log(1/\delta))$, but
which have 
much simpler analyses. We describe
our simpler constructions in \Section{body},
and our better
constructions
in \Section{random-hashes}. We show
in \Section{tight} that our analyses of the required sparsity in our
schemes are tight up to a constant factor.
In \Section{linalg} we discuss how our new schemes speed up
the numerical linear algebra algorithms in \cite{CW09} for 
approximate linear regression and best rank-$k$ approximation in the
streaming model of computation. We also show in \Section{linalg} 
that a wide range of JL distributions automatically provides sketches
for approximate
matrix product as defined in \cite{Sarlos06}.  While
\cite{Sarlos06} also
showed this, it lost a logarithmic factor in the target dimension due
to a union bound in its reduction; the work of \cite{CW09} avoided
this loss, but only for the JL distribution of random
sign matrices.  We show a simple and general reduction which incurs no
loss in
parameters. Plugging in our
sparse JL transform then yields faster linear algebra algorithms using
the same space.
In \Section{open} we state two open problems for future work.

\subsection{Our Approach}
Our constructions are depicted in
\Figure{pictures}. \Figure{pictures}(a) represents the DKS construction of
\cite{DKS10} in which each item is hashed to $s$ random target
coordinates with replacement. Our two schemes achieving $s =
\Theta(\eps^{-1}\log(1/\delta))$ are as follows. Construction (b) is
much like (a) except that we hash coordinates $s$ times 
  {\em without} replacement; we call this the {\it graph construction}, since hash locations are specified by a bipartite graph with $d$ left vertices, $k$ right vertices, and left-degree $s$.  In (c),
the target vector is
divided into $s$ contiguous blocks each of equal size $k/s$, and a
given coordinate in the original vector is hashed to a random location
in each block
(essentially this is the \textsc{CountSketch} of
\cite{CCF04}, 
though we use a higher degree of independence in our hash
functions); we call this the {\it block construction}. 
In all cases (a), (b), and (c), we randomly flip the sign of a
coordinate in the original vector and divide by $\sqrt{s}$ before
adding it in any location in the target vector. 

We give two different analyses for both our constructions (b) and
(c).  
Since we consider linear embeddings, without loss of generality we can
assume $\|x\|_2 = 1$, in which case the JL lemma follows by showing that $\|Sx\|_2^2 \in [(1-\eps)^2, (1+\eps)^2]$, which is implied by $|\|Sx\|_2^2 - 1| \le 2\eps - \eps^2$. Thus it suffices to show that for
any unit norm $x$,
\begin{equation}
\Pr_S(|\|Sx\|_2^2 - 1| > 2\eps - \eps^2) < \delta .\EquationName{suffices}
\end{equation}
We furthermore observe that both our graph and block constructions have the property that the entries of our embedding matrix $S$ can be written as
\begin{equation}
S_{i,j} = \eta_{i,j}\sigma_{i,j}/\sqrt{s} , \EquationName{matrix-entries}
\end{equation}
where the $\sigma_{i,j}$ are independent and uniform in $\{-1,1\}$, and $\eta_{i,j}$ is an indicator random variable for the event $S_{i,j}\neq 0$ (in fact in our analyses we will only need that the $\sigma_{i,j}$ are $O(\log(1/\delta))$-wise independent). Note that the $\eta_{i,j}$ are not independent, since in both constructions we have that there are exactly $s$ non-zero entries per column. Furthermore in the block construction, knowing that $\eta_{i,j} = 1$ for $j$ in some block implies that $\eta_{i,j'} = 0$ for all other $j'$ in the same block.

To outline our analyses, look at the random variable 
\begin{equation}
Z \eqdef \|Sx\|_2^2 - 1 = \frac 1s\cdot \sum_{r=1}^k\sum_{i\neq j\in[d]} \eta_{r,i}\eta_{r,j}\sigma_{r,i}\sigma_{r,j}x_i x_j . \EquationName{Zdef}
\end{equation}
Our proofs all use Markov's bound on the $\ell$th moment $Z^\ell$ to give $\Pr(|Z| > 2\eps - \eps^2) < (2\eps - \eps^2)^{-\ell} \cdot \E Z^\ell$ for $\ell =
\log(1/\delta)$ an even integer. The task is then to bound
$\E Z^\ell$. 
In our first
approach, we observe that $Z$ is a quadratic form in the $\sigma_{i,j}$ of \Equation{matrix-entries},
and thus its moments can be bounded via the Hanson-Wright inequality
\cite{HW71}.  This analysis turns out to reveal that the hashing to
coordinates in the target vector need not be done randomly, but can in
fact be specified by any sufficiently good code (i.e.\ the $\eta_{i,j}$ need not be random).
Specifically, it suffices that for any $j\neq j'\in [d]$, $\sum_{i=1}^k \eta_{i,j}\eta_{i,j'} = O(s^2/k)$. That is, no two columns have their non-zero entries in more than $O(s^2/k)$ of the same rows.
In (b), this translates to the columns of the embedding
matrix (ignoring the random signs and division by $\sqrt{s}$) to be
codewords in a constant-weight binary code of weight $s$ and minimum
distance $2s - O(s^2/k)$.  In (c), if for each $j\in[d]$ we
let $C_j$ be a length-$s$ vector with entries in $[k/s]$ specifying
where coordinate $j$ is mapped to in each block, it suffices for
$\{C_j\}_{j=1}^d$ to be a code of minimum distance
$s - O(s^2/k)$. It is fairly easy to see that if one wants
a deterministic hash function,
it is necessary for the columns of the embedding matrix to be
specified by a code: if two coordinates have their non-zeroes in many of the same rows,
it means those coordinates collide often.  Since collision is the source of error,
an adversary in this case could ask to embed a vector which has its
mass equally spread on these two coordinates, causing large error with large probability
over the choice of random signs.  What our analysis shows
is that not only is a good code necessary, but it is also
sufficient.

In our second analysis approach, we define
\begin{equation}
Z_r = \sum_{i\neq j\in[d]} \eta_{r,i}\eta_{r,j}\sigma_{r,i}\sigma_{r,j}x_i x_j . \EquationName{Zrdef}
\end{equation} 
so that
\begin{equation}
Z = \frac 1s\sum_{r=1}^k Z_r \EquationName{ZZr} .
\end{equation}
We show that to bound $\E Z^\ell$ it suffices to bound $\E Z_r^t$ for each $r\in[k], t\in [\ell]$. To bound $\E Z_r^t$, we expand
expand $Z_r^t$ to obtain a polynomial with roughly $d^{2t}$
terms. We view its monomials as being in correspondence with
graphs, group monomials that map to the same graph, then do some
combinatorics to make the expectation calculation feasible.  
We remark that a similar tactic of mapping monomials to graphs then carrying out combinatorial arguments is frequently used to analyze the eigenvalue spectrum of random matrices; see for example work of Wigner \cite{Wigner55}, or the work of F\"{u}redi and Koml\'{o}s \cite{FK81}. In our approach here, we assume that the random signs as well as the hashing to coordinates in the target vector are
done $O(\log(1/\delta))$-wise independently. This combinatorial approach of mapping to graphs played
a large role in our previous analysis of the DKS construction \cite{KN10}, as well as a later analysis of that construction in
\cite{BOR10}. 

We point out here that
\Figure{pictures}(c) is somewhat
simpler to implement, since there are simple constructions of
$O(\log(1/\delta))$-wise hash families \cite{CW79}.
\Figure{pictures}(b) on the other hand requires hashing without
replacement, which amounts to using random permutations and
can be derandomized using almost
$O(\log(1/\delta))$-wise
independent permutation families \cite{KNR09} (see \Remark{derand-perm}).


\section{Conventions and Notation}\SectionName{notation}
\begin{definition}
For $A\in\R^{n\times n}$,
the {\em Frobenius norm} of $A$ is
$\|A\|_F =
\sqrt{\sum_{i,j} A_{i,j}^2}$.
\end{definition}

\begin{definition}
For $A\in\R^{n\times n}$, the {\em operator norm} of $A$ is
$\|A\|_2 = \sup_{\|x\|_2 = 1} \|Ax\|_2$.
In the case $A$ is symmetric, this is also the
largest magnitude of an eigenvalue of $A$.
\end{definition}

Henceforth, all logarithms are base-$2$ unless explicitly
stated otherwise. For a positive integer $n$ we use $[n]$ to
denote the set $\{1,\ldots,n\}$.  
We will always be focused on embedding
a vector $x\in\R^d$ into $\R^k$, and we
assume $\|x\|_2 = 1$
without loss of generality (since our embeddings are linear).
All vectors $v$ are assumed to be column
vectors, and $v^T$ denotes its transpose.
We often implicitly assume that various quantities, such as $1/\delta$, are powers
of $2$ or $4$, which is without loss of generality. 
Space complexity bounds (as in \Section{linalg}), are always measured
in bits.

\section{Code-Based Constructions}\SectionName{body}
In this section, we provide analyses of our constructions (b) and (c)
in \Figure{pictures} when the non-zero entry locations are deterministic but satisfy a certain condition. In particular, in the analysis in this section we assume that for any $i\neq j\in [d]$, 
\begin{equation}
\sum_{r=1}^k \eta_{r,i}\eta_{r,j} = O(s^2/k) . \EquationName{few-collisions}
\end{equation}
That is, no two columns have their non-zero entries in more than $O(s^2/k)$ of the same rows. We show how to use error-correcting codes to ensure \Equation{few-collisions} in \Remark{code} for the block construction, and in \Remark{easy-perm} for the graph construction. Unfortunately this step will require setting $s$ to be slightly larger than the desired $O(\eps^{-1}\log(1/\delta))$. We give an alternate analysis in \Section{random-hashes} which avoids assuming \Equation{few-collisions} and obtains an improved bound for $s$ by not using deterministic $\eta_{r,i}$.

In what follows, we assume $k = C\cdot \eps^{-2}\log(1/\delta)$ for a sufficiently large
constant $C$, and that $s$ is some integer dividing $k$ satisfying $s \ge
2(2\eps - \eps^2)^{-1}\log(1/\delta) = \Theta(\eps^{-1}\log(1/\delta))$. We also assume that the $\sigma_{i,j}$ are $2\ell$-wise independent for $\ell = \log(1/\delta)$, so that $\E (\|Sx\|_2^2 - 1)^\ell$ is fully determined.

\paragraph{Analysis of \Figure{pictures}(b) and \Figure{pictures}(c) code-based constructions:}
Recall from \Equation{Zdef}
\begin{equation*}
Z \eqdef \|Sx\|_2^2 - 1 = \frac 1s \sum_{r=1}^k\sum_{i\neq j\in[d]}\eta_{r,i}\eta_{r,j}\sigma_{r,i}\sigma_{r,j}x_ix_j .
\end{equation*}
Note $Z$ is a quadratic form in
$\sigma$ which can be written as $\sigma^TT\sigma$ for a
$kd\times kd$ block-diagonal matrix $T$. There are $k$ blocks, each
$d\times d$, where in the $r$th block $T_r$ we have $(T_r)_{i,j} =
\eta_{r,i}\eta_{r,j}x_ix_j/s$ for $i\neq j$ and $(T_r)_{i,i} = 0$ for
all $i$. Now, $\Pr(|Z| > 2\eps-\eps^2) = \Pr(|\sigma^T T \sigma| > 2\eps-\eps^2)$.  
To obtain an upper bound for this probability, we use the Hanson-Wright inequality
combined with a Markov bound.

\begin{theorem}[Hanson-Wright inequality {\cite{HW71}}]\TheoremName{dkn}
Let $z = (z_1,\ldots,z_n)$ be a vector of i.i.d.\ Rademacher $\pm 1$
random variables.  For any symmetric $B\in\R^{n\times n}$ and
$\ell\ge 2$,
$$ \E\left|z^TBz - \mathrm{trace}(B)\right|^\ell \le
C^\ell\cdot
\max\left\{\sqrt{\ell}\cdot \|B\|_F, \ell\cdot \|B\|_2\right\}^\ell
$$
for some universal constant $C>0$ independent of $B,n,\ell$.
\end{theorem}

We prove our construction satisfies the JL lemma by
applying
\Theorem{dkn} with $z=\sigma, B=T$.

\begin{lemma}\LemmaName{frobenius}
$\|T\|_F^2 = O(1/k)$.
\end{lemma}
\begin{proof}
$$ \|T\|_F^2 = \frac{1}{s^2}\cdot\sum_{i\neq j\in [d]}x_i^2x_j^2\cdot
\left(\sum_{r=1}^k\eta_{r,i}\eta_{r,j}\right) \le
O(1/k)\cdot \sum_{i\neq j\in[d]}x_i^2x_j^2 \le O(1/k)\cdot \|x\|_2^4 = O(1/k) ,$$
where the first inequality used \Equation{few-collisions}.
\end{proof}

\begin{lemma}\LemmaName{operator}
$\|T\|_2 \le 1/s$.
\end{lemma}
\begin{proof}
Since $T$ is block-diagonal, its eigenvalues are the eigenvalues of
each block. For a block $T_r$, write $T_r = (1/s)\cdot (S_r - D_r)$.
$D_r$ is diagonal with $(D_r)_{i,i} = \eta_{r,i}x_i^2$, and $(S_r)_{i,j} =
\eta_{r,i}\eta_{r,j}x_ix_j$.  Since $S_r$ and $D_r$ are
both positive semidefinite, we have $\|T\|_2 \le (1/s)\cdot
\max\{\|S_r\|_2, \|D_r\|_2\}$.  We have $\|D_r\|_2 \le \|x\|_{\infty}^2 \le
1$.  Define $u\in\R^d$ by $u_i = \eta_{r,i}x_i$ so $S_r = uu^T$. Thus
$\|S_r\|_2 = \|u\|_2^2 \le \|x\|_2^2 = 1$.
\end{proof}

By \Equation{suffices}, it now suffices to prove the following theorem.

\begin{theorem}\TheoremName{thisisit}
$\Pr_{\sigma}(|Z| > 2\eps - \eps^2) < \delta$.
\end{theorem}
\begin{proof}
By a Markov bound applied to $Z^\ell$ for $\ell$ an even integer,
$$\Pr_\sigma(|Z| > 2\eps - \eps^2) < (2\eps - \eps^2)^{-\ell}\cdot \E_\sigma Z^\ell .$$
Since $Z = \sigma^TT\sigma$ and $\mathrm{trace}(T) = 0$,
applying \Theorem{dkn} with $B=T$, $z = \sigma$, and $\ell =
\log(1/\delta)$ gives
\begin{equation}\EquationName{apply-spectral}
\Pr_\sigma(|Z| > \eps) < C^\ell\cdot
\max\left\{O(\eps^{-1})\cdot \sqrt\frac{\ell}{k},
  (2\eps - \eps^2)^{-1}\frac{\ell}{s}\right\}^\ell .
\end{equation}
since the $\ell$th moment is determined by $2\log(1/\delta)$-wise
independence of $\sigma$. We conclude the proof by noting that the
expression in \Equation{apply-spectral}
is at most $\delta$ for our choices for $s,k,\ell$.
\end{proof}

We now discuss how to choose the non-zero locations in $S$ to ensure \Equation{few-collisions}.

\begin{remark}\RemarkName{code}
\textup{
Consider the block construction, and for $i\in [d]$ let $C_i\in [k/s]^s$ specify the locations of the non-zero entries for column $i$ of $S$ in each of the $s$ blocks. Then \Equation{few-collisions} is equivalent to  $\mathcal{C} = \{C_1,\ldots,C_d\}$ being an error-correcting code with relative distance $1 - O(s/k)$, i.e.\ that no $C_i,C_j$ pair for $i\neq j$ agree in more than $O(s^2/k)$ coordinates. It is thus important to know whether such a code exists. Let $h:[d]\times[s]\rightarrow [k/s]$ be such that $h(i,r)$ gives the non-zero location in block $r$ for column $i$, i.e.\ $(C_i)_r = h(i,r)$. Note that having relative distance $1 - O(s/k)$ is to say that for every $i\neq j\in[d]$, $h(i,r) = h(j,r)$ for at most $O(s^2/k)$ values of $r$. For $r\in [s]$ let $X_r$ be an indicator random variable for the event $h(i,r) = h(j,r)$, and define $X = \sum_{r=1}^s X_r$. Then $\E X = s^2/k$, and if $s^2/k = \Omega(\log(d/\delta))$, then a Chernoff bound shows that $X = O(s^2/k)$ with probability at least $1-\delta/d^2$ over the choice of $h$ (in fact it suffices to use Markov's bound applied to the $O(\log(d/\delta))^{th}$ moment implied by the Chernoff bound so that $h$ can be $O(\log(d/\delta))$-wise independent, but we do not dwell on this issue here since \Section{random-hashes} obtains better parameters). Thus by a union bound over all $\binom{d}{2}$ pairs $i\neq j$, $\mathcal{C}$ is a code with the desired properties with probability at least $1-\delta/2$. Note that the condition $s^2/k = \Omega(\log(d/\delta))$ is equivalent to $s = \Omega(\eps^{-1}\sqrt{\log(d/\delta)\log(1/\delta)})$.
We also point out that we may assume without loss of generality that
$d = O(\eps^{-2}/\delta)$.  This is because there exists an
embedding into this dimension with sparsity $1$ using only $4$-wise
independence with distortion
$(1+\eps)$ and success probability $1-\delta/2$ \cite{CCF04,
  ThorupZhang12}.
It is worth noting
that in the construction in this section, potentially $h$
could be deterministic given an explicit code with our desired
parameters. 
}
\end{remark}

\begin{remark}\RemarkName{easy-perm}
\textup{It is also possible to use a code to specify the hash
  locations in the graph construction.  In particular, 
  let the $j$th entry of the $i$th column of the embedding matrix be
  the $j$th symbol of the $i$th codeword (which we call $h(i,j)$) in a
  weight-$s$ binary code of minimum distance $2s - O(s^2/k)$ for $s \ge
  2\eps^{-1}\log(1/\delta)$. Define
  $\eta_{i,j,r}$ for $i,j\in[d], r\in[s]$ as an indicator variable for
  $h(i,r) = h(j,r) = 1$.  Then, the error is again exactly as in
  \Equation{Zdef}. Also, as in
  \Remark{code}, such a code
  can be shown to exist via the probabilistic method (the Chernoff
  bound can be applied using negative dependence, followed by a union
  bound) as long as $s =
  \Omega(\eps^{-1}\sqrt{\log(d/\delta)\log(1/\delta)})$. We omit the
  details since \Section{random-hashes} obtains better parameters.
}
\end{remark}


\begin{remark}\RemarkName{random-code}
\textup{
Only using \Equation{few-collisions}, it is impossible to improve our sparsity bound further.
For example, consider an instantiation of the block construction in which \Equation{few-collisions} is satisfied. Create a new set of $\eta_{r,i}$ which change only in the case $r=1$ so that $\eta_{1,i} = 1$ for all $i$, so that \Equation{few-collisions} still holds.
In our construction this corresponds to all
indices colliding in the first chunk of $k/s$ coordinates, which
creates an error term of $(1/s) \cdot \sum_{i\neq j} x_ix_j\sigma_{r,i}\sigma_{r,j}$.  Now, suppose $x$ consists of $t =
(1/2)\cdot\log(1/\delta)$ entries each with value $1/\sqrt{t}$.  Then, with
probability $\sqrt{\delta} \gg \delta$, all these entries receive the
same sign under $\sigma$ and contribute a total error of $\Omega(t/s)$
in the first chunk alone.  We thus need $t/s = O(\eps)$, which implies
$s = \Omega(\eps^{-1}\log(1/\delta))$.
}
\end{remark}

\section{Random Hashing Constructions}\SectionName{random-hashes}
In this section, we show that if the hash functions $h$ described in
\Remark{code} and \Remark{easy-perm} are not specified by fixed
codes, but
rather are chosen at random from some family of sufficiently high
independence, then one can achieve sparsity
$O(\eps^{-1}\log(1/\delta))$ (in the case of \Figure{pictures}(b), we
actually need almost k-wise independent {\em permutations}).
Recall our bottleneck in reducing the sparsity in \Section{body} was
actually obtaining the codes,
discussed in \Remark{code} and \Remark{easy-perm}.


We perform our analysis by bounding the $\ell^{th}$ moment of $Z = \|Sx\|_2^2 - 1$ from first principles for $\ell = \Theta(\log(1/\delta))$ an even integer (for this particular scheme, it seems the Hanson-Wright inequality does not simplify any details of the proof). To show \Equation{suffices} we then use Markov's inequality to say $\Pr(|Z|>\lambda) < \lambda^{-\ell}\cdot \E Z^{\ell}$. 
Although the $\eta_{i,j}$ are specified differently in the two constructions, in both cases they are easily seen to be {\it negatively correlated}; that is, for any subset $T\subseteq [k]\times [d]$ (in fact in our proof we will only be concerned with $|T| \le \ell$) we have $\E\prod_{(i,j)\in T} \eta_{i,j} \le (s/k)^{|T|}$. Also, each construction has $\sum_{i=1}^k \eta_{i,j} = s$ with probability $1$ for all $j\in [d]$, and thus, recalling the definition of $Z_r$ from \Equation{Zrdef},
$$Z = \frac 1 s\cdot \sum_{r=1}^k\sum_{i\neq j\in [d]} x_i x_j \sigma_{r,i}\sigma_{r,j} \eta_{r,i}\eta_{r,j} = \frac 1s \cdot \sum_{r=1}^k Z_r .$$

We first bound the $t^{th}$ moment of each $Z_r$ for $1\le t \le \ell$. As in the Frobenius norm moment bound of \cite{KN10}, and also used later in \cite{BOR10},
the main idea is to construct a correspondence between the monomials appearing in $Z_r^t$ and certain graphs. Notice
\begin{equation}
Z_r^t = \sum_{\substack{i_1,\ldots,i_t,j_1,\ldots,j_t\in[d]\\i_1\neq j_1, \ldots,i_t\neq j_t}}\prod_{u=1}^t \eta_{r,i_u}\eta_{r,j_u} x_{i_u}x_{j_u} \sigma_{r,i_u}\sigma_{r,j_u} . \EquationName{expand-it}
\end{equation}
To each monomial above we associate a directed multigraph with labeled edges whose vertices correspond to the distinct $i_u$ and $j_u$. An $x_{i_u}x_{j_u}$ term corresponds to a directed edge with label $u$ from the vertex corresponding to $i_u$ to the vertex corresponding to $j_u$. The basic idea we use to bound $\E Z_r^t$ is to group these monomials based on their associated graphs.

\begin{figure*}
\begin{center}
\includegraphics[scale=0.85]{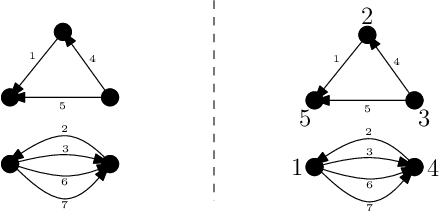} 
\caption{Example of a graph in $\mathcal{G}_t$ on the left with $v=5$, $t = 7$ and $j_1=j_5, i_1 = j_4, i_4 = i_5, j_2 = i_3 = i_6 = i_7, i_2 = j_3 = j_6 = j_7$. Example graph with the same restrictions on the right, but in $\mathcal{G}'_t$.}\FigureName{graph_example}
\end{center}
\end{figure*}

\begin{lemma}\LemmaName{random-graphs}
For $t>1$ an integer, $\E_{\eta,\sigma} Z_r^t \le t (2e^2)^t\cdot \begin{cases} (s/k)^2 \ & t < 2\ln(k/s) \\ (t/\ln(k/s))^t \ &\mathrm{otherwise} \end{cases}$.
\end{lemma}
\begin{proof}
We have
\begin{equation}
\E_{\eta,\sigma} Z_r^t = \sum_{\substack{i_1,\ldots,i_t,j_1,\ldots,j_t\in[d]\\i_1\neq j_1, \ldots,i_t\neq j_t}} \left(\prod_{u=1}^t x_{i_u}x_{j_u}\right) \cdot \left(\E_\sigma \prod_{u=1}^t \sigma_{r,i_u}\sigma_{r,j_u}\right) \cdot \left(\E_\eta\prod_{u=1}^t \eta_{r,i_u}\eta_{r,j_u}\right) .\EquationName{firstsum}
\end{equation}

Define $\mathcal{G}_t$ as the set of directed multigraphs with $t$ edges having distinct labels in $[t]$ and no self-loops, with between $2$ and $t$ vertices (inclusive), and where every vertex has non-zero and even degree (we use degree to denote the sum of in- and out-degrees). Let $f$ map variable sequences to their corresponding graph. That is, we draw a directed edge labeled $u$ from the vertex representing $i_u$ to that representing $j_u$ for $u = 1,\ldots,t$, where one vertex represents all the $i_u,j_u$ which are assigned the same element of $[d]$ (see \Figure{graph_example}). For a graph $G$, let $v$ be its number of vertices, and let $d_u$ be the degree of vertex $u$. By construction every monomial maps to a graph with $t$ edges. Also we need only consider graphs with all even vertex degrees since a monomial whose graph has at least one vertex with odd degree will have at least one random sign $\sigma_{i,r_u}$ appearing an odd number of times and thus have expectation zero. Then,
\allowdisplaybreaks
\begin{align}
\nonumber \E_{\eta,\sigma} Z_r^t & =
\sum_{G\in\mathcal{G}_t}\sum_{\substack{i_1\neq
    j_1,\ldots,i_t \neq
    j_t\in[d]\\f((i_u,j_u)_{u=1}^t)
  = G}} \left(\prod_{u=1}^t x_{i_u}x_{j_u}\right) \cdot \E_\eta \prod_{u=1}^t
\eta_{r,i_u}\eta_{r,j_u} \\
{}&= \sum_{G\in\mathcal{G}_t}\sum_{\substack{i_1\neq
    j_1,\ldots,i_t \neq
    j_t\in[d]\\f((i_u,j_u)_{u=1}^t)
  = G}} \left(\prod_{u=1}^t x_{i_u}x_{j_u}\right) \cdot \left(\frac
sk\right)^v \EquationName{neg-cor}\\
{}&\le \sum_{G\in\mathcal{G}_t}
\left(\frac sk\right)^v\cdot
v! \cdot
\frac{1}{\binom{t}{d_1/2,\ldots,d_v/2}} \EquationName{binom}\\
{}&= \sum_{G\in\mathcal{G}_t'}
\left(\frac sk\right)^v\cdot
\frac{1}{\binom{t}{d_1/2,\ldots,d_v/2}} \EquationName{useprime}\\
{}&\le (e/2)^t \cdot \sum_{v=2}^t \left(\frac sk\right)^v \cdot
\frac{1}{t^t}\cdot \left(
\sum_{G\in\mathcal{G}'_t}
\prod_{u=1}^v \sqrt{d_u}^{d_u}\right)\EquationName{lastg},
\end{align}
where $\mathcal{G}'_t$ is the set of all directed multigraphs as in $\mathcal{G}_t$, but in which vertices are labeled as well, with distinct labels in $[v]$ (see \Figure{graph_example}; the vertex labels can be arbitrarily permuted).

\Equation{neg-cor} used that $\eta_{r,1},\ldots,\eta_{r,d}$ are independent for any $r$.
For \Equation{binom}, note that $(\|x\|_2^2)^t = 1$, and the coefficient of $\prod_{u=1}^v x_{a_u}^{d_u}$ in its expansion for $\sum_{u=1}^v d_u = 2t$ is $\binom{t}{d_1/2,\ldots,d_v/2}$. Meanwhile, the coefficient of this monomial when summing over all $i_1\neq j_1, \ldots, i_t\neq j_t$ for a particular $G\in\mathcal{G}_t$ is at most $v!$. For \Equation{useprime}, we move from graphs in $\mathcal{G}_t$ to those in $\mathcal{G}_t'$, and for any $G\in\mathcal{G}_t$ there are exactly $v!$ ways to label vertices. This is because for any graph $G\in\mathcal{G}_t$ there is a canonical way of labeling the vertices as $1,\ldots,v$ since there are no isolated vertices. Namely, the vertices can be labeled in increasing order of when they are first visited by an edge when processing edges in order of increasing label (if two vertices are both visited for the first time simultaneously by some edge, then we can break ties consistently using the direction of the edge). Thus the vertices are all identified by this canonical labeling, implying that the $v!$ vertex labelings all give distinct graphs in $\mathcal{G}'_t$. \Equation{lastg} follows since $t! \ge t^t/e^t$ and 
$$\prod_{u=1}^v (d_u/2)! \le \prod_{u=1}^v 2^{-d_u/2} \sqrt{d_u}^{d_u} = 2^{-\sum_{u=1}^v d_u/2} \prod_{u=1}^v \sqrt{d_u}^{d_u} = 2^{-t} \prod_{u=1}^v \sqrt{d_u}^{d_u}. $$

The summation over $G$ in \Equation{lastg} is over
the $G\in\mathcal{G}_t'$ with $v$ vertices. Let us bound this summation for some fixed choice of vertex degrees $d_1,\ldots,d_v$. For any given $i$, consider the set of all graphs $\mathcal{G}''_i$ on $v$ labeled vertices with distinct labels in $[v]$, and with $i$ edges with distinct labels in $[i]$ (that is, we do not require even edge degrees, and some vertices may even have degree $0$). For a graph $G\in\mathcal{G}''_i$, let $d_u'$ represent the degree of vertex $u$ in $G$. For $a_1,\ldots,a_v>0$ define the function
\begin{equation}
S_i(a_1,\ldots,a_v) = \sum_{G\in\mathcal{G}''_i} \prod_{u=1}^v \sqrt{a_u}^{d'_u} .
\EquationName{induct-contribution}
\end{equation}
Let $\mathcal{G'}_t(d_1,\ldots,d_v)$ be those graphs $G\in\mathcal{G}'_t$ with $v$ vertices such that vertex $u$ has degree $d_u$. Then
$$ \sum_{\substack{G\in\mathcal{G'}_t(d_1,\ldots,d_v)}} \prod_{u=1}^v \sqrt{d_u}^{d_u} \le S_t(d_1,\ldots,d_v) $$
since $\mathcal{G}'_t(d_1,\ldots,d_v) \subset \mathcal{G}''_t$. To upper bound $S_t(a_1,\ldots,a_v)$, note $S_0(a_1,\ldots,a_v) = 1$. For $i>1$, note any graph in $\mathcal{G''}_i$ can be formed by taking a graph $G\in\mathcal{G''}_{i-1}$ and adding an edge labeled $i$ from $u$ to $w$ for some vertices $u\neq w$ in $G$. This change causes $d_u',d_w'$ to both increase by $1$, whereas all other degrees stay the same. Thus considering \Equation{induct-contribution},
$$ S_{i+1}(a_1,\ldots,a_v)/S_i(a_1,\ldots,a_v) \le \left(\sum_{u\neq w\in[v]} \sqrt{a_u}\cdot \sqrt{a_w}\right)
 \le \left(\sum_{u=1}^v \sqrt{a_u}\right)^2 \le \left(\sum_{u=1}^v a_u\right)\cdot v ,$$
with the last inequality using Cauchy-Schwarz. Thus by induction, $S_t(a_1,\ldots,a_v) \le (\sum_{u=1}^v a_u)^t\cdot v^t$. Since $\sum_{u=1}^v d_u = 2t$, we have $S_t(d_1,\ldots,d_v) \le (2tv)^t$. We then have that the summation in \Equation{lastg} is at most the number of choices of even $d_1,\ldots,d_v$ summing to $2t$ (there are $\binom{t-1}{v-1}<2^t$ such choices), times $(2tv)^t$, implying
$$ \E_{\eta,\sigma} Z_r^t \le (2e)^t\cdot \sum_{v=2}^t \left(\frac sk\right)^v \cdot
v^t .$$
By differentiation, the quantity $(s/k)^v v^t$ is maximized for $v =
\max\left\{2, t/\ln(k/s)\right\}$ (recall $v \ge 2$),
giving our lemma.
\end{proof}

\begin{corollary}\CorollaryName{simplified-graph}
For $t>1$ an integer, $\E_{\eta,\sigma} Z_r^t \le t (2e^3)^t (s/k)^2  t^t$.
\end{corollary}
\begin{proof}
We use \Lemma{random-graphs}. In the case $t<2\ln(k/s)$ we can multiply the $(s/k)^2$ term by $t^t$ and still obtain an upper bound, and in the case of larger $t$ we have $(t/\ln(k/s))^t \le t^t$ since $k\ge s$. Also when $t \ge 2\ln(k/s)$ we have $e^t (s/k)^2 \ge 1$, so that $t(2e^2)^t t^t \le t(2e^3)^t (s/k)^2 t^t$.
\end{proof}

\begin{theorem}\TheoremName{graph-thm}
For some $s\in\Theta(\eps^{-1}\log(1/\delta)), k \in \Theta(\eps^{-2}\log(1/\delta))$, we have $\Pr_{h,\sigma}(|Z| > 2\eps - \eps^2) < \delta$.
\end{theorem}
\begin{proof}
We choose $\ell$ an even integer to be specified later. Using \Equation{ZZr} and $\E Z_r = 0$ for all $r$,
\allowdisplaybreaks
\begin{align}
\nonumber \E Z^\ell &= \frac{1}{s^\ell}\cdot
\sum_{q=1}^{\ell/2} \sum_{\substack{r_1<\ldots<r_q\in[k]\\\ell_1,\ldots,\ell_q\\ \forall i\
    \ell_i > 1\\ \sum_i \ell_i = \ell}}
\binom{\ell}{\ell_1,\ldots,\ell_q} \cdot \E \prod_{i=1}^q
Z_{r_i}^{\ell_i}\\
{}&\le  \frac{1}{s^\ell}\cdot
\sum_{q=1}^{\ell/2} \sum_{\substack{r_1<\ldots<r_q\in[k]\\\ell_1,\ldots,\ell_q\\ \forall i\
    \ell_i > 1\\ \sum_i \ell_i = \ell}}
\binom{\ell}{\ell_1,\ldots,\ell_q} \cdot \prod_{i=1}^q
\E Z_{r_i}^{\ell_i} \EquationName{neg-cor2}\\
{}&\le \frac 1{s^\ell} \sum_{q=1}^{\ell/2} \sum_{\substack{r_1<\ldots<r_q\in[k]\\\ell_1,\ldots,\ell_q\\ \forall i\ \ell_i > 1\\ \sum_i \ell_i = \ell}} \frac{\ell!}{\prod_{i=1}^q \ell_i!} \cdot \left(\prod_{i=1}^q \ell_i\right)\cdot (2e^3)^\ell \cdot \left(\frac sk\right)^{2q} \cdot \prod_{i=1}^q \ell_i^{\ell_i} \EquationName{use-cor}\\
{}&\le \frac 1{s^\ell} \sum_{q=1}^{\ell/2} \sum_{\substack{r_1<\ldots<r_q\in[k]\\\ell_1,\ldots,\ell_q\\ \forall i\ \ell_i > 1\\ \sum_i \ell_i = \ell}} e^{-q}\cdot \ell! \cdot \left(\prod_{i=1}^q \ell_i\right)\cdot (2e^4)^\ell \cdot \left(\frac sk\right)^{2q}\EquationName{simplify1}\\
{}&\le \frac 1{s^\ell} \sum_{q=1}^{\ell/2} \sum_{\substack{r_1<\ldots<r_q\in[k]\\\ell_1,\ldots,\ell_q\\ \forall i\ \ell_i > 1\\ \sum_i \ell_i = \ell}} e^{-q}\cdot \ell! \cdot (4e^4)^\ell \cdot \left(\frac sk\right)^{2q}\EquationName{simplify2}\\
{}&\le \left(\frac{4e^3(\ell+1)}s\right)^\ell\cdot (\ell+1)\cdot \sum_{q=1}^{\ell/2} \sum_{\substack{r_1<\ldots<r_q\in[k]\\\ell_1,\ldots,\ell_q\\ \forall i\ \ell_i > 1\\ \sum_i \ell_i = \ell}} e^{-q}\cdot \left(\frac sk\right)^{2q}\EquationName{simplify3}\\
{}&\le \left(\frac{8e^3(\ell+1)}s\right)^\ell\cdot (\ell+1)\cdot \sum_{q=1}^{\ell/2}
e^{-q}\cdot \binom{k}{q}\cdot \left(\frac sk\right)^{2q}\EquationName{middle-calc} \\
{}&\le \left(\frac{8e^3(\ell+1)}s\right)^\ell\cdot (\ell+1)\cdot \sum_{q=1}^{\ell/2}
\left(\frac {s^2}{qk}\right)^{q}\EquationName{last-expr}
\end{align}

\Equation{neg-cor2} follows since the expansion of $\prod_i Z_{r_i}^{\ell_i}$ into monomials contains all nonnegative terms, in which the participating $\eta_{r,i}$ terms are negatively correlated, and thus $\E \prod_i Z_{r_i}^{\ell_i}$ is term-by-term dominated when expanding into a sum of monomials by the case when the $\eta_{r,i}$ are independent. \Equation{use-cor} uses \Corollary{simplified-graph}, and \Equation{simplify1} uses $\ell_i! \ge e(\ell_i/e)^{\ell_i}$. \Equation{simplify2} compares geometric and arithmetic means, giving $\prod_{i=1}^q \ell_i \le (\sum_{i=1}^q \ell_i/q)^q \le (\ell/q)^q \le \binom{\ell}{q} < 2^\ell$. \Equation{simplify3} bounds $\ell! \le (\ell+1)\cdot ((\ell+1)/e)^\ell$. \Equation{middle-calc} follows since there are $\binom{k}{q}$ ways to choose the $r_i$, and there are at most $2^{\ell - 1}$ ways to choose the $\ell_i$ summing to $\ell$.  Taking derivatives shows that the right hand side of \Equation{last-expr} is maximized for $q = \max\{1,s^2/(ek)\}$, which will be bigger than $1$ and less than $\ell/2$ by our choices of $s,k,\ell$ that will soon be specified. Then $q = s^2/(ek)$ gives a summand of $e^q \le e^{\ell/2}$.  We choose $\ell \ge \ln(\delta^{-1}(\ell+1)\ell/2) = \Theta(\log(1/\delta))$ and $s \ge 8e^4\sqrt{e}(\ell+1)/(2\eps - \eps^2) = \Theta(\eps^{-1}\log(1/\delta))$ so that \Equation{last-expr} is at most $(2\eps-\eps^2)^\ell\cdot \delta$. Then to ensure $s^2/(ek) \le \ell/2$ we choose $k = 2s^2/(e\ell) = \Theta(\eps^{-2}\log(1/\delta))$. The theorem then follows by Markov's inequality.
\end{proof}

\begin{remark}\RemarkName{derand-perm}
\textup{
In order to use fewer random bits to sample from the graph construction, we can use the following implementation. We realize the distribution over $S$ via two hash functions
$h:[d]\times [k]\rightarrow \{0,1\}$ and $\sigma:[d]\times[s]\rightarrow
\{-1,1\}$. The function $\sigma$ is drawn from from a $2\log(1/\delta)$-wise
independent family.  The function $h$ has the property that for any
$i$, exactly $s$ distinct $r\in[k]$ have $h(i,r) = 1$; in particular, we
pick $d$ seeds $\log(1/\delta)$-wise independently to determine
$h_i$ for $i=1,\ldots,d$, and where each $h_i$ is drawn
from a $\gamma$-almost $2\log(1/\delta)$-wise independent family of
permutations on $[d]$ for $\gamma = (\eps s/(d^2
k))^{\Theta(\log(1/\delta))}$. The seed length required for any one
such permutation is $O(\log(1/\delta)\log d + \log(1/\gamma)) =
O(\log(1/\delta)\log d)$
\cite{KNR09}, and thus we can pick $d$ such seeds
$2\log(1/\delta)$-wise independently using total seed length
$O(\log^2(1/\delta)\log d)$. We then let $h(i,r) = 1$ iff some
$j\in[s]$ has $h_i(j) = r$.
Recall that a $\gamma$-almost
$\ell$-wise independent family of permutations from $[d]$ onto itself
is a family of permutations $\mathcal{F}$ where the image
of any fixed $\ell$ elements in $[d]$ has statistical distance at most
$\gamma$ when choosing a random $f\in\mathcal{F}$ when compared with
choosing a uniformly random permutation $f$. Now, there are
$(kd^2)^{\ell}$ monomials in the expansion of $Z^\ell$. In each such
monomial, the coefficient of the $\E \prod_u h(i_u,r_u) h(j_u,r_u)$
term is at most $s^{-\ell}$.  In the end, we want $\E_{h,\sigma} Z^\ell <
O(\eps)^{\ell}$ to apply Markov's inequality.  Thus, we want
$(kd^2/s)^{\ell} \cdot \gamma < O(\eps)^{\ell}$.
}
\end{remark}

\begin{remark}\RemarkName{manyeps}
\textup{
  It is worth noting that if one wants distortion $1\pm\eps_i$ with
  probability $1-\delta_i$ simultaneously for all $i$ in some set
  $S$, our proof of \Theorem{graph-thm} 
  reveals that it suffices
  to set $s = C\cdot \sup_{i\in S} \eps_i^{-1}\log(1/\delta_i)$ and $k
  = C\cdot \sup_{i\in S} \eps_i^{-2}\log(1/\delta_i)$.
}
\end{remark}

\section{Tightness of analyses}\SectionName{tight}
In this section we show that sparsity
$\Omega(\eps^{-1}\log(1/\delta))$ is required in \Figure{pictures}(b)
and \Figure{pictures}(c), even if the hash functions used are
completely random. We also show that sparsity
$\tilde{\Omega}(\eps^{-1}\log^2(1/\delta))$ is required in the DKS
construction (\Figure{pictures}(a)), nearly matching the upper bounds
of \cite{BOR10,KN10}.
Interestingly, all three of our proofs of (near-)tightness of analyses
for these three constructions use the same hard input vectors.  In
particular, if $s = o(1/\eps)$, then we show that a vector with $t =
\floor{1/(s\eps)}$ entries each of value $1/\sqrt{t}$ incurs large
distortion with large probability.  If $s = \Omega(1/\eps)$ but is
still not sufficiently large, we show that the vector
$(1/\sqrt{2},1/\sqrt{2},0,\ldots,0)$ incurs large distortion with
large probability (in fact, for the DKS scheme one can even take the
vector $(1,0,\ldots,0)$).

\subsection{Near-tightness for DKS Construction}\SectionName{dks-tight}
The main theorem of this section is the following.

\begin{theorem}\TheoremName{dks-tight}
The DKS construction of \cite{DKS10} requires sparsity $s =
\Omega(\eps^{-1}\cdot \ceil{\log^2(1/\delta)/\log^2(1/\eps)})$ to
achieve
distortion $1\pm\eps$ with success probability $1 - \delta$.
\end{theorem}

Before proving \Theorem{dks-tight}, we recall the DKS
construction (\Figure{pictures}(a)).  First, we replicate
each coordinate $s$ times while
preserving the $\ell_2$ norm.  That
is, we produce the vector $\tilde{x} =
(x_1,\ldots,x_1,x_2,\ldots,x_2,\ldots,x_d,\ldots,x_d)/\sqrt{s}$, where
each $x_i$ is replicated $s$ times.  Then, pick a random $k\times
ds$ embedding matrix $A$ for $k = C\eps^{-2}\log(1/\delta)$ where each
column has exactly one non-zero entry, in a location defined by
some random function $h:[ds]\rightarrow [k]$, and where this non-zero
entry is $\pm 1$, determined by some random function
$\sigma:[ds]\rightarrow \{-1,1\}$.  The value $C>0$ is some fixed
constant. The final embedding is $A$ applied to $\tilde{x}$.  We are
now ready to prove \Theorem{dks-tight}. The proof is similar to that
of \Theorem{tight-scheme}.

Our proof will use the following standard fact.

\begin{fact}[{\cite[Proposition B.3]{MR95}}]\FactName{mr95}
For all $t,n\in\R$ with $n\ge 1$ and $|t| \le n$,
$$e^t (1 - t^2/n) \le (1 + t/n)^n \le e^t .$$
\end{fact}

\bigskip

\begin{proofof}{\Theorem{dks-tight}}
First suppose $s \le 1/(2\eps)$. Consider a vector with $t =
\floor{1/(s\eps)}$ non-zero coordinates each of value $1/\sqrt{t}$.
If there is exactly one pair $\{i,j\}$ that collides under $h$, and
furthermore the signs agree under $\sigma$, the $\ell_2$ norm squared
of our embedded vector will be $(st-2)/(st) + 4/(st)$.  Since $1/(st)
\ge \eps$, this quantity is at least $1+2\eps$.  The event of exactly
one pair $\{i,j\}$ colliding occurs with probability
\begin{align*}
 \binom{st}{2} \cdot \frac 1k \cdot \prod_{i=0}^{st-2}(1 - i/k) &
\ge  \Omega\left(\frac{1}{\log(1/\delta)}\right) \cdot (1 -
\eps/2)^{1/\eps}\\
&{} =\Omega(1/\log(1/\delta)) ,
\end{align*}
which is much larger than $\delta/2$ for $\delta$ smaller than some
constant.  Now,
given a collision, the colliding items have the same sign with
probability $1/2$.

We next consider the case $1/(2\eps) < s \le 4/\eps$. Consider the
vector $x = (1,0,\ldots,0)$.  If there are exactly three pairs
$\{i_1,j_1\},\ldots, \{i_3,j_3\}$ that collide under $h$ in three
distinct target coodinates, and furthermore the signs agree under
$\sigma$, the $\ell_2$ norm squared of our embedded vector will be
$(s - 6)/(s) + 12/(s) > 1 + 3\eps/2$. The event of three pairs
colliding occurs with probability
\begin{align*}
 \binom{s}{2}\binom{s-2}{2}\binom{s-4}{2} \cdot \frac{1}{3!}\cdot
 \frac{1}{k^3}\cdot
 \prod_{i=0}^{s-4}(1 - i/k) &
\ge  \Omega\left(\frac{1}{\log^3(1/\delta)}\right) \cdot (1 -
\eps/8)^{4/\eps}\\
&{}= \Omega(1/\log^3(1/\delta)) ,
\end{align*}
which is much larger than $\delta/2$ for $\delta$ smaller than some
constant.  Now,
given a collision, the colliding items have the same sign with
probability $1/8$.

We lastly consider the case $4/\eps < s \le
2c\eps^{-1}\log^2(1/\delta)/\log^2(1/\eps)$ for some constant
$c>0$ (depending on $C$) to be determined later.  First note this case
only exists when $\delta = O(\eps)$.
Define $x = (1,0,\ldots,0)$. Suppose there exists
an integer $q$ so that
\begin{enumerate}
\item $q^2/s \ge 4\eps$
\item $q/s < \eps$
\item $(s/(qk))^q(1 - 1/k)^s > \delta^{1/3}$.
\end{enumerate}

First we show it is possible to satisfy the above conditions
simultaneously for our range of $s$. We set $q = 2\sqrt{\eps s}$,
satisfying
item 1 trivially, and item 2 since $s > 4/\eps$.
For item 3, \Fact{mr95} gives
$$(s/(qk))^q \cdot (1 - 1/k)^s \ge \left(\frac{s}{qk}\right)^q \cdot
e^{-s/k} \cdot \left(1 -\frac{s}{k^2}\right) .$$
The $e^{-s/k} \cdot (1 - (s/k^2))$ term is at least $\delta^{1/6}$ by
the settings of $s,k$, and the $(s/(qk))^q$ term is also at least
$\delta^{1/6}$ for $c$ sufficiently small. 

Now, consider the event $\mathcal{E}$ that exactly $q$ of the $s$
copies of $x_1$ are hashed to $1$ by $h$, and to $+1$ by
$\sigma$.  
If $\mathcal{E}$ occurs,
then coordinate $1$ in the target vector contributes $q^2/s \ge 4\eps$
to $\ell_2^2$ in the target vector by item 1 above, whereas these
coordinates only contribute $q/s < \eps$ to $\|x\|_2^2$ by item 2
above, thus causing error at least $3\eps$.
Furthermore, the $s-q$ coordinates which do not hash to $1$ are being
hashed to a vector of length $k - 1 = \omega(1/\eps^2)$ with random
signs, and thus these coordinates have their $\ell_2^2$ contribution
preserved up to $1\pm o(\eps)$ with constant probability by
Chebyshev's inequality.  It thus just remains to show that
$\Pr(\mathcal{E}) \gg \delta$.
We have
\begin{align*}
\Pr(\mathcal{E}) &= \binom{s}{q} \cdot k^{-q} \cdot \left(1 - \frac
  1k\right)^{s - q}\cdot 1/2^q\\
{}&\ge \left(\frac{s}{qk}\right)^q \cdot \left(1 - \frac
  1k\right)^{s}\cdot \frac {1}{2^q}\\
{}& > \delta^{1/3} \cdot \frac{1}{2^q}.
\end{align*}

The $2^{-q}$ term is 
$\omega(\delta^{1/3})$ and thus
overall $\Pr(\mathcal{E}) = \omega(\delta^{2/3}) \gg \delta$.
\end{proofof}

\subsection{Tightness of \Figure{pictures}(b) analysis}
\begin{theorem}\TheoremName{tight-other-scheme}
For $\delta$ smaller than a constant depending on
$C$ for $k = C\eps^{-2}\log(1/\delta)$, the graph construction
of \Section{random-hashes} requires
$s = \Omega(\eps^{-1}\log(1/\delta))$ to obtain  distortion $1\pm\eps$
with probability $1-\delta$.
\end{theorem}
\begin{proof}
First suppose $s \le 1/(2\eps)$. We consider a vector with $t =
\floor{1/(s\eps)}$
non-zero coordinates each of value $1/\sqrt{t}$.
If there is exactly one set $i,j,r$ with $i\neq j$ such that $S_{r,
  i}, S_{r,j}$ are both non-zero for the embedding matrix $S$
(i.e., there is exactly one collision), then the total error is
$2/(ts) \ge 2\eps$.  It just remains to show that this happens with
probability larger than $\delta$. The probability
of this occurring is
\begin{align*}
s^2 \cdot \binom{t}{2}\cdot \frac 1k\cdot 
\frac{k - s}{k - 1} \cdots \frac{k-2s+2}{k-s+1}\cdot
\left(\frac{(k-2s+1)!}{(k-ts+1)!}\right)\cdot
\left(\frac{(k-s)!}{k!}\right)^{t-2} & \ge \frac{s^2t^2}{2k}\cdot
\left(\frac{k-st}{k}\right)^{st}\\
&{}\ge \frac{s^2t^2}{2k}\cdot
\left(1 - \frac{s^2t^2}{k}\right)\\
&{}= \Omega(1/\log(1/\delta)) .
\end{align*}

Now consider the case $1/(2\eps) < s < c\cdot \eps^{-1}\log(1/\delta)$
for some small constant $c$. Consider the vector
$(1/\sqrt{2},1/\sqrt{2},0,\ldots,0)$.
Suppose there are exactly
$2s\eps$ collisions, i.e. $2s\eps$ distinct values of $r$ such that
$S_{r,i}, S_{j,r}$ are both non-zero
(to avoid tedium we disregard floors and ceilings and
just assume $s\eps$ is an integer). Also, suppose that in each
colliding row $r$ we
have $\sigma(1, r) = \sigma(2, r)$.
Then, the total error would be $2\eps$.  It just
remains to show that this happens with probability larger than
$\delta$. The probability of signs agreeing in exactly $2\eps s$
chunks is $2^{-2\eps s} > 2^{-2c\log(1/\delta)}$, which is larger than
$\sqrt{\delta}$ for $c < 1/4$. The probability of exactly $2\eps s$
collisions is
\begin{align}
\nonumber \binom{s}{2\eps s} \cdot \left(\prod_{i=0}^{2\eps s - 1}\frac{s - i}{k
  - i}\right)\cdot \left(\prod_{i=0}^{s - 2\eps s - 1}\frac{k - i - s}{k -
  i -2\eps s}\right) &\ge 
\left(\frac{1}{2\eps}\right)^{2\eps s}\cdot
\left(\frac{(1-2\eps) s}{k}\right)^{2\eps s}\cdot \left(1 -
  \frac{s}{k-s}\right)^{s - 2\eps s}\\
&{}\ge \left(\frac{s}{4\eps k}\right)^{2\eps s}\cdot \left(1 -
  \frac{2s}{k}\right)^s .\EquationName{reusable}
\end{align}

It suffices for the right hand side to be at least $\sqrt{\delta}$
since $h$ is independent of $\sigma$, and thus the total probability
of error larger than $2\eps$ would be greater than $\sqrt{\delta}^2 =
\delta$. Taking natural logarithms, it suffices to have
$$ 2\eps s\ln\left(\frac{4\eps k}{s}\right) - s\ln\left(1 -
  \frac {2s}k\right) \le \ln(1/\delta)/2 . $$
Writing $s = q/\eps$ and $a = 4C\log(1/\delta)$, the left hand side is
$2q\ln(a/q) + \Theta(s^2/k)$. 
Taking a derivative shows $2q\ln(a/q)$ is monotonically increasing for
$q < a/e$. Thus as long as $q < ca$ for a sufficiently small constant
$c$, $2q\ln(a/q) < \ln(1/\delta)/4$.  Also, the $\Theta(s^2/k)$ term
is at most $\ln(1/\delta)/4$ for $c$ sufficiently small.
\end{proof}

\subsection{Tightness of \Figure{pictures}(c) analysis}
\begin{theorem}\TheoremName{tight-scheme}
For $\delta$ smaller than a constant depending on
$C$ for $k = C\eps^{-2}\log(1/\delta)$, the block construction
of \Section{random-hashes} requires
$s = \Omega(\eps^{-1}\log(1/\delta))$ to obtain  distortion $1\pm\eps$
with probability $1-\delta$.
\end{theorem}
\begin{proof}
First suppose $s \le 1/(2\eps)$. Consider a vector with $t =
\floor{1/(s\eps)}$
non-zero coordinates each of value $1/\sqrt{t}$.
If there is exactly one set $i,j,r$ with $i\neq j$ such that $h(i, r)
= h(j, r)$ (i.e. exactly one collision), then the total error is
$2/(ts) \ge 2\eps$.  It just remains to show that this happens with
probability larger than $\delta$.

The probability of exactly one collision is
\begin{align}
\nonumber s\cdot \left[\frac{t!\cdot \binom{k/s}{t}}{(k/s)^t}\right]^{s-1} \cdot
\binom{t}{2} \cdot \left(\frac ks\right)\cdot \left[\frac{(t-2)!\cdot
    \binom{k/s - 1}{t - 2}}{(k/s)^t}\right] & \ge s\cdot \left(1 -
  \frac{st}{k}\right)^{t(s-1)} \cdot \binom{t}{2}\cdot \left(\frac
  sk\right)\left(1 - \frac{st}{k}\right)^{t-2}\\
\nonumber {}&= \frac{s^2 t(t-1)}{2k}\cdot \left(1 - \frac{st}{k}\right)^{st-2}\\
\nonumber {}&\ge \frac{s^2 t(t-1)}{2k}\cdot \left(1 -
  \frac{s^2t^2}{k}\right)\\
\nonumber {}&= \Omega(1/\log(1/\delta)) ,
\end{align}
which is larger than $\delta$ for $\delta$ smaller than a
universal constant.
 
Now consider $1/(2\eps) < s < c\cdot \eps^{-1}\log(1/\delta)$ for some
small constant $c$.  Consider the vector $x =
(1/\sqrt{2},1/\sqrt{2},0,\ldots,0)$.  Suppose there are exactly
$2s\eps$ collisions, i.e. $2s\eps$ distinct values of $r$ such that
$h(1,r) = h(2,r)$ (to avoid tedium we disregard floors and ceilings and
just assume $s\eps$ is an integer). Also, suppose that in each
colliding chunk $r$ we
have $\sigma(1, r) = \sigma(2, r)$.
Then, the total error would be $2\eps$.  It just
remains to show that this happens with probability larger than
$\delta$. The probability of signs agreeing in exactly $2\eps s$
chunks is $2^{-2\eps s} > 2^{-2c\log(1/\delta)}$, which is larger than
$\sqrt{\delta}$ for $c < 1/4$. The probability of exactly $2\eps s$
collisions is
$$
\binom{s}{2\eps s} \left(\frac sk\right)^{2\eps s}\left(1 - \frac
  sk\right)^{(1 - 2\eps)s} \ge \left(\frac{s}{2\eps k}\right)^{2\eps
  s} \left(1 - \frac sk\right)^{(1 - 2\eps)s}
$$

The above is at most $\sqrt{\delta}$, by the analysis following
\Equation{reusable}.  
Since $h$ is independent of $\sigma$, the total probability
of having error larger than $2\eps$ is greater than $\sqrt{\delta}^2 =
\delta$.
\end{proof}

\section{Faster numerical linear algebra streaming
  algorithms}\SectionName{linalg}
The works of \cite{CW09,Sarlos06} gave algorithms to solve various
approximate numerical linear algebra problems given small memory and a
only one or few passes over an input matrix. They considered models where
one only sees a row or column at a time of some matrix
$A\in\R^{d\times n}$. Another update model considered was
the turnstile streaming model. In this model,
the matrix $A$ starts off as the all zeroes matrix.  One
then sees a sequence of $m$ updates
$(i_1,j_1,v_1),\ldots,(i_m,j_m,v_m)$, where each update $(i,j,v)$
triggers the change $A_{i,j} \leftarrow A_{i,j} + v$. The goal in all
these models is to compute some functions of $A$ at the end of seeing
all rows, columns, or turnstile updates.  The algorithm should
use little memory (much less than what is
required to store $A$ explicitly). Both works \cite{CW09,Sarlos06}
solved problems such as approximate linear regression and best
rank-$k$ approximation by reducing to the problem of sketches for 
  approximate matrix products.  
Before delving further, first we give a definition.

\begin{definition}
Distribution $\mathcal{D}$ over $\R^{k\times d}$ has {\em
  $(\eps,\delta,\ell)$-JL moments} if for all $x$ with $\|x\|_2 = 1$,
$$ \E_{S\sim\mathcal{D}}\left|\|Sx\|_2^2 -
    1\right|^\ell \le \eps^\ell\cdot \delta .$$
\end{definition}

Now, the following theorem is a generalization of 
\cite[Theorem 2.1]{CW09}. The theorem states that any distribution
with JL moments
also provides a sketch for approximate matrix products. A
similar statement was made in \cite[Lemma 6]{Sarlos06}, but that
statement was slightly weaker in its parameters because it resorted to
a union bound, which we avoid by using Minkowski's inequality. 

\begin{theorem}\TheoremName{improved-matrix}
Given $\eps,\delta\in(0,1/2)$,
let $\mathcal{D}$ be any distribution over matrices with $d$ columns
with the $(\eps,\delta,\ell)$-JL moment property for some $\ell\ge 2$.
Then for $A,B$ any real matrices with $d$ rows,
$$
\Pr_{S\sim\mathcal{D}}\left(\|A^TS^TSB - A^TB\|_F  > 3\eps\|A\|_F\|B\|_F\right)
< \delta .
$$
\end{theorem}
\begin{proof}
Let $x,y\in\R^d$ each have $\ell_2$ norm $1$.  Then 
$$\inprod{Sx, Sy} = \frac{\|Sx\|_2^2 + \|Sy\|_2^2 -
  \|S(x-y)\|_2^2}{2} $$
so that, defining $\|X\|_p = (\E |X|^p)^{1/p}$ (which is a norm for $p\ge 1$ by Minkowski's inequality),
\begin{align*}
\|\inprod{Sx, Sy} - \inprod{x, y}\|_\ell &=
\frac 12 \cdot \left\|(\|Sx\|_2^2 - 1) +
      (\|Sy\|_2^2 - 1) -
  (\|S(x-y)\|_2^2 - \|x-y\|_2^2)\right\|_\ell \\
&{} \le \frac 12\cdot\left(\left\|\|Sx\|_2^2 - 1\right\|_\ell + \left\|\|Sy\|_2^2 - 1\right\|_\ell +
  \left\|\|S(x-y)\|_2^2 - \|x-y\|_2^2\right\|_\ell\right)\\
&{} \le \frac 12\cdot \left(\eps\cdot \delta^{1/\ell} + \eps\cdot\delta^{1/\ell} + \|x -y\|_2^2 \cdot \eps\cdot\delta^{1/\ell}\right)\\
&{} \le 3\eps\cdot \delta^{1/\ell}
\end{align*}
Now,
if $A$ has $n$ columns and $B$ has $m$ columns, label the columns of
$A$ as $x_1,\ldots,x_n\in\R^d$ and the columns of $B$ as
$y_1,\ldots,y_m\in\R^d$.
Define the random variable $X_{i,j} = 1/(\|x_i\|_2\|y_j\|_2)\cdot
(\inprod{Sx_i, Sy_j} -
\inprod{x_i, y_j})$.
Then $\|A^TS^TSB - A^TB\|_F^2 = \sum_{i=1}^n\sum_{j=1}^m \|x_i\|_2^2\cdot
\|y_j\|_2^2\cdot X_{i,j}^2$. Then again by Minkowski's inequality since $\ell/2 \ge 1$,
{
\allowdisplaybreaks
\begin{align*}
 \left\|\|A^TS^TSB - A^TB\|_F^2\right\|_{\ell/2} &=
 \left\|\sum_{i=1}^n\sum_{j=1}^m \|x_i\|_2^2
     \cdot \|y_j\|_2^2\cdot X_{i,j}^2\right\|_{\ell/2}\\
{}&\le \sum_{i=1}^n\sum_{j=1}^m \|x_i\|_2^2\cdot \|y_j\|_2^2 \cdot
\|X_{i,j}^2\|_{\ell/2}\\
{}&= \sum_{i=1}^n\sum_{j=1}^m \|x_i\|_2^2\cdot \|y_j\|_2^2 \cdot
\|X_{i,j}\|_{\ell}^2\\
{}&\le (3\eps \delta^{1/\ell})^2\cdot \left(\sum_{i=1}^n\sum_{j=1}^m\|x_i\|_2^2\cdot \|y_j\|_2^2\right)\\
{}&= (3\eps \delta^{1/\ell})^2\cdot \|A\|_F^2\|B\|_F^2
\end{align*}
}

Then by Markov's inequality and using $\E \|A^TS^TSB - A^TB\|_F^\ell = \|\|A^TS^TSB - A^TB\|_F^2\|_{\ell/2}^{\ell/2}$, 
$$ \Pr\left(\|A^TS^TSB - A^TB\|_F > 3\eps\|A\|_F\|B\|_F\right) \le \left(\frac{1}{3\eps\|A\|_F\|B\|_F}\right)^\ell\cdot \E\|A^TS^TSB - A^TB\|_F^\ell \le \delta .$$
\end{proof}

\begin{remark}
\textup{
Often when one constructs a JL distribution $\mathcal{D}$ over
$k\times d$ matrices, it is shown that for all $x$ with $\|x\|_2=1$
and for all $\eps>0$,
$$
\Pr_{S\sim
  \mathcal{D}}\left(\left|\|Sx\|_2^2 - 1\right| > \eps\right) <
e^{-\Omega(\eps^2 k + \eps k)} .
$$
Any such distribution
automatically satisfies the $(\eps,e^{-\Omega(\eps^2 k + \eps k)}, \min\{\eps^2 k, \eps k\})$-JL moment
property for any $\eps>0$
by converting the tail bound into a moment bound via integration by parts.
}
\end{remark}

\begin{remark}\RemarkName{ose}
\textup{
After this work there was interest in finding sparse {\em oblivious subspace embeddings}, i.e.\ a randomized and sparse $S\in\R^{k\times n}$ such that for any $U\in\R^{n\times d}$ with orthonormal columns, $\Pr(\|(SU)^T(SU) - I\| > \eps) < \delta$. Here the norm is $\ell_2$ to $\ell_2$ operator norm, and thus $\|(SU)^T(SU) - I\| \le \eps$ implies that $(1-\eps)\|x\|_2^2 \le \|Sx\|_2^2 \le (1+\eps)\|x\|_2^2$ for all $x$ in the column span of $U$. It was shown in \cite{CW13,MM13,NN13b} that such $S$ exists with one non-zero entry per column and $k = O(d^2/(\eps^2\delta))$ rows. It has sinced been pointed out to us by Huy L\^{e} Nguy$\tilde{\hat{\mbox{e}}}$n that this result also follows from \Theorem{improved-matrix}. Indeed, \cite{ThorupZhang12} provides a distribution with $(\eps',\delta,2)$-JL moments with $k = O(\eps'^{-2}\delta^{-1})$ rows, and supported on matrices each with exactly one non-zero entry per column. The claim then follows by applying \Theorem{improved-matrix} with $A=B=U$ and $\eps' = \eps/(3d)$ by noting that $\|U\|_F = \sqrt{d}$ and that operator norm is upper bounded by Frobenius norm.
}
\end{remark}

Now we arrive at the main point of this section. Several algorithms
for approximate linear regression and best
rank-$k$ approximation in \cite{CW09} simply maintain $SA$ as $A$ is
updated, where $S$ comes from the JL distribution with
$\Omega(\log(1/\delta))$-wise independent $\pm 1/\sqrt{k}$ entries.
In fact though, their analyses of their algorithms only use the fact
that this distribution
satisfies the approximate matrix product sketch guarantees of
\Theorem{improved-matrix}. Due to \Theorem{improved-matrix} though, we
know that {\em any} distribution satisfying the
$(\eps,\delta)$-JL moment condition gives an
approximate matrix product sketch. Thus, random Bernoulli matrices may
be replaced with our sparse JL distributions
in this work. We now state some of the algorithmic results
given in \cite{CW09} and describe how our constructions provide
improvements in the update time (the time to process new columns,
rows, or turnstile updates).

As in \cite{CW09}, when stating our results we will ignore the space
and time complexities of storing and evaluating the hash functions
in our JL distributions.  We discuss this issue later in
\Remark{hashing}.

\subsection{Linear regression} 
In this problem we have an $A\in\R^{d\times n}$
and $b\in\R^d$. We would like to compute
a vector $\tilde{x}$ such that
$\|A\tilde{x} - b\|_F \le (1+\eps)\cdot \min_{x^*}\|Ax^* -
   b\|_F$ with probability $1-\delta$.  In \cite{CW09}, it is assumed
   that the entries of $A,b$
   require $O(\log(nd))$ bits of precision to store precisely. 
   Both $A,b$ receive turnstile updates.

Theorem 3.2 of \cite{CW09} proves
that such an $\tilde{x}$ can
be computed with probability $1-\delta$ from $SA$ and $Sb$, where $S$
is drawn from a distribution that simultaneously satisfies both the
$(1/2,\eta^{-r}\delta)$ and 
$(\sqrt{\eps/r},\delta)$-JL moment properties for some fixed
constant $\eta>1$ in their proof, and where $\mathrm{rank}(A) \le r \le n$. Thus due
to \Remark{manyeps}, we have the following.

\begin{theorem}\TheoremName{regression-onepass}
There is a one-pass streaming algorithm for linear regression in the
turnstile model where one maintains a sketch of size
$O(n^2\eps^{-1}\log(1/\delta)\log(nd))$.  Processing each update
requires $O(n + \sqrt{n/\eps}\cdot \log(1/\delta))$ arithmetic operations
and hash function evaluations. 
\end{theorem}

\Theorem{regression-onepass} improves the update complexity of
\cite{CW09}, which was $O(n\eps^{-1}\log(1/\delta))$.

\subsection{Low rank approximation}
In this problem, we have an $A\in\R^{d\times n}$ of rank $\rho$ with
entries that require precision $O(\log(nd))$ to store.  We would like
to compute the best rank-$r$ approximation $A_r$ to $A$. We define
$\Delta_r \eqdef \|A - A_r\|_F$ as the error of $A_r$.  We relax the
problem by only requiring that we compute a matrix $A'_r$ such that
$\|A - A'_r\|_F \le (1+\eps)\Delta_r$ with probability $1-\delta$ over
the randomness of the algorithm.

\paragraph{Two-pass algorithm:}

Theorem 4.4 of \cite{CW09} gives
a $2$-pass algorithm where in the first pass, one maintains
$SA$ where $S$
is drawn from a distribution that simultaneously satisfies both the
$(1/2,\eta^{-r}\delta)$ and 
$(\sqrt{\eps/r},\delta)$-JL moment properties for some fixed constant $\eta>1$ in their proof.  
It is also assumed that $\rho \ge 2r+1$.
The first pass is thus
sped up again as in \Theorem{regression-onepass}.

\paragraph{One-pass algorithm for column/row-wise updates:}

Theorem 4.5 of \cite{CW09} gives a one-pass algorithm in the case that
$A$ is seen either one whole column or row at a time. The
algorithm maintains both $SA$ and $SAA^T$ where $S$
is drawn from a distribution that simultaneously satisfies both the
$(1/2,\eta^{-r}\delta)$ and 
$(\sqrt{\eps/r},\delta)$-JL moment properties. This implies the
following.

\begin{theorem}\TheoremName{onepass-column}
There is a one-pass streaming algorithm for approximate low rank
approximation with row/column-wise updates
 where one maintains a sketch of size
$O(r\eps^{-1}(n+d)\log(1/\delta)\log(nd))$.  Processing each update
requires $O(r + \sqrt{r/\eps}\cdot \log(1/\delta))$ amortized arithmetic
operations
and hash function evaluations per entry of $A$. 
\end{theorem}

\Theorem{onepass-column} improves the amortized update complexity of
\cite{CW09}, which was $O(r\eps^{-1}\log(1/\delta))$.

\paragraph{Three-pass algorithm for row-wise updates:}

Theorem 4.6 of \cite{CW09} gives a three-pass algorithm using less
space in the case that $A$ is seen one row at a time.  Again, the
first pass simply maintains $SA$ where $S$ is drawn from a
distribution that satisfies both the
$(1/2,\eta^{-r}\delta)$ and 
$(\sqrt{\eps/r},\delta)$-JL moment properties. This pass is sped up
using our sparser JL distribution.

\paragraph{One-pass algorithm in the turnstile model, bi-criteria:}
Theorem 4.7 of \cite{CW09} gives a one-pass algorithm under turnstile
updates where $SA$ and $RA^T$ are maintained in the stream.  $S$
is drawn from a distribution satisfying both the $(1/2,
\eta^{-r\log(1/\delta)/\eps}\delta)$ and
$(\eps/\sqrt{r\log(1/\delta)},\delta)$-JL moment
properties. $R$ is drawn from a distribution satisfying both the
$(1/2, \eta^{-r}\delta)$ and
$(\sqrt{\eps/r}, \delta)$-JL moment properties. Theorem
4.7 of \cite{CW09} then shows how to compute a matrix of rank
$O(r\eps^{-1}\log(1/\delta))$ which achieves the desired error
guarantee given $SA$ and $RA^T$.

\paragraph{One-pass algorithm in the turnstile model:}
Theorem 4.9 of \cite{CW09} gives a one-pass algorithm under turnstile
updates where $SA$ and $RA^T$ are maintained in the stream.  $S$
is drawn from a distribution satisfying both the $(1/2,
\eta^{-r\log(1/\delta)/\eps^2}\delta)$ and
$(\eps\sqrt{\eps/(r\log(1/\delta))},\delta)$-JL moment
properties. $R$ is drawn from a distribution satisfying both the
$(1/2, \eta^{-r}\delta)$ and
$(\sqrt{\eps/r}, \delta)$-JL moment properties. Theorem
4.9 of \cite{CW09} then shows how to compute a matrix of rank
$r$ which achieves the desired error guarantee given $SA$ and $RA^T$.

\begin{remark}\RemarkName{hashing}
\textup{
In the algorithms above, we counted the number of hash function
evaluations that must be performed. 
We use our construction in \Figure{pictures}(c), which uses
$2\log(1/\delta)$-wise independent hash functions.
Standard constructions of $t$-wise
independent hash functions over universes with elements fitting in a
machine word require $O(t)$ time to evaluate \cite{CW79}.  In our
case, this
would blow up our update time by factors such as $n$ or $r$, which could be
large.  Instead, we use fast multipoint evaluation of polynomials. The
standard construction \cite{CW79} of our desired hash functions
mapping some domain $[z]$ onto itself for $z$ a power of $2$ takes a
degree-$(t-1)$ polynomial $p$ with random coefficients in $\mathbb{F}_z$.
The hash function evaluation at some point $y$ is then the evaluation
$p(y)$ over $\mathbb{F}_z$. \Theorem{fastmult} below states that $p$
can be evaluated at $t$ points in total time $\tilde{O}(t)$. We note
that in the theorems above, we are always required to evaluate some
$t$-wise independent hash function on many more than $t$ points per
stream update.  Thus, we can group these evaluation points into
groups of size $t$ then perform fast multipoint evaluation for each
group. We borrow this idea from \cite{KNPW11}, which used 
it to
give a fast algorithm for moment estimation in data streams.
}
\end{remark}

\begin{theorem}[{\cite[Ch. 10]{GG99}}]\TheoremName{fastmult}
Let $\mathbf{R}$ be a ring, and let $q\in \mathbf{R}[x]$ be a
degree-$t$ polynomial. Then, given distinct
$x_1,\ldots,x_t\in\mathbf{R}$, all the values $q(x_1),\ldots,q(x_t)$
can be computed using $O(t\log^2t\log\log t)$ operations
over $\mathbf{R}$.
\end{theorem}

\section{Open Problems}\SectionName{open}
In this section we state two explicit open problems. For the first, observe that our graph construction is quite similar to a sparse JL construction of Achlioptas \cite{Achlioptas03}. The work of \cite{Achlioptas03} proposes a random normalized sign matrix where each column has an {\em expected} number $s$ of non-zero entries, so that in the notation of this work, the $\eta_{i,j}$ are i.i.d.\ Bernoulli with expectation $s/k$. Using this construction, \cite{Achlioptas03} was able to achieve $s = k/3$ without causing $k$ to increase over analyses of dense constructions, even by a constant factor. Meanwhile, our graph construction requires that there be exactly $s$ non-zero entries per column. This sole change was the reason we were able to obtain better asymptotic bounds on the sparsity of $S$ in this work, but in fact we conjecture an even stronger benefit than just asymptotic improvement. The first open problem is to resolve the following conjecture.

\begin{conjecture}
Fix a positive integer $k$. For $x\in\R^d$, define $Z_{x,s}^{A}$ as the error random variable $|\|Sx\|_2^2 - \|x\|_2^2|$ when $S$ is the sparse construction of \cite{Achlioptas03} with sparsity parameter $s$.  Let $Z_{x,s}^{G}$ be similarly defined, but when using our graph construction. Then for any $x\in\R^d$ and any $s\in[k]$, $Z_{x,s}^A$ stochastically dominates $Z_{x,s}^G$. That is, for all $x\in\R^d$, $s\in[k]$, $\lambda > 0$, $\Pr(Z_{x,s}^A > \lambda) \ge \Pr(Z_{x,s}^G > \lambda)$.
\end{conjecture}

A positive resolution of this conjecture would imply that not only does our graph construction obtain better {\em asymptotic} performance than \cite{Achlioptas03}, but in fact obtains stronger performance in a very definitive sense.

The second open problem is the following. Recall that the ``metric Johnson-Lindenstrauss lemma'' \cite{JL84} states that for any $n$ vectors in $\R^d$, there is a linear map into $\R^k$ for $k = O(\eps^{-2}\log n)$ which preserves all pairwise Euclidean distances of the $n$ vectors up to $1\pm\eps$. \Lemma{jl-lemma} implies this metric JL lemma by setting $\delta < 1/\binom{n}{2}$ then performing a union bound over all $\binom{n}{2}$ pairwise difference vectors. Alon showed that $k = \Omega(\eps^{-2}\log n / \log(1/\eps))$ is necessary \cite{Alon03}. Our work shows that metric JL is also achievable where every column of the embedding matrix has at most $s = O(\eps^{-1}\log n)$ non-zeroes, and this is also known to be tight up to an $O(\log(1/\eps))$ factor \cite{NN13}. Thus, for metric JL, the lower bounds for both $k$ and $s$ are off by $O(\log(1/\eps))$ factors. Meanwhile, for the form of the JL lemma in \Lemma{jl-lemma} where one wants to succeed on any fixed vector with probability $1-\delta$ (the ``distributional JL lemma''), the tight lower bound on $k$ of $\Omega(\eps^{-2}\log(1/\delta))$ is known \cite{JW13,KMN11}. Thus it seems that obtaining lower bounds for distributional JL is an easier task.

\bigskip

\noindent \textbf{Question:} Can we obtain a tight lower bound of $s = \Omega(\eps^{-1}\log(1/\delta))$ for distributional JL in the case that $k = O(\eps^{-2}\log(1/\delta)) < d/2$, thus removing the $O(\log(1/\eps))$ factor gap?

\section*{Acknowledgments}
We thank Venkat Chandar, Venkatesan Guruswami, Swastik Kopparty, and
Madhu Sudan for useful discussions about error-correcting codes,
David Woodruff for answering several questions about
\cite{CW09}, Piotr Indyk and Eric Price for useful comments
and discussion, and Mark Rudelson and Dan Spielman for both pointing out the similarity of our proof of \Lemma{random-graphs} to the types of arguments that are frequently used to analyze the eigenvalue spectrum of random matrices. We thank Huy L\^{e} Nguy$\tilde{\hat{\mbox{e}}}$n for pointing out \Remark{ose}. We also thank the anonymous referees for many helpful comments.

\bibliographystyle{plain}

\bibliography{../allpapers}

\end{document}